\documentclass[sigconf]{acmart}

\setcopyright{none}
\renewcommand\footnotetextcopyrightpermission[1]{} 
\usepackage[utf8]{inputenc}%
\usepackage{algorithm}
\usepackage[noend]{algpseudocode} 
\algrenewcommand\algorithmicrequire{\textbf{Input:}}
\algrenewcommand\algorithmicensure{\textbf{Output:}}

\usepackage[hyphenbreaks]{breakurl}

\usepackage{color}

\allowdisplaybreaks
\usepackage{csquotes} %
\PassOptionsToPackage{hyphens}{url}
\usepackage{float}
\usepackage{subcaption}
\usepackage{array}
\usepackage{paralist}
\usepackage{tcolorbox}
\usepackage{amsmath}  

\usepackage{tabulary}  
\usepackage{multirow}
\usepackage{amssymb}
\usepackage[n, advantage, operators, sets, adversary, landau, probability, notions, logic, ff,mm, primitives, events, complexity, asymptotics, keys]{cryptocode}
\usepackage{amsthm}

\newtheorem{theorem}{Theorem}[section]

\newtheorem{lemma}[theorem]{Lemma}

\makeatletter
\def\url@leostyle{%
  \@ifundefined{selectfont}{\def\UrlFont{\small}}%
  {\def\UrlFont{}}%
}
\makeatother
\renewcommand{\footnoterule}{%
  \kern -3pt
  \hrule width 1in 
  \kern 2pt
}

\renewcommand{\footnotesize}{\fontsize{8}{9}\selectfont}

\usepackage{amsmath}
\usepackage{amsfonts}
\usepackage{color, colortbl}
\usepackage[hyphens]{url}
\definecolor{darkgreen}{RGB}{135,0,40}

\newcommand{\descr}[1]{\smallskip\noindent\textbf{#1}}

\newif\ifcomment
\commenttrue 

\ifcomment
	\newcommand{\edc}[1]{\textbf{\em\color{red}#1}}
\else  
    \newcommand\edc[1]{}
\fi

\ifcomment
	\newcommand{\bo}[1]{\textbf{\em\color{blue}#1}}
\else  
    \newcommand\bo[1]{}
\fi

\ifcomment
	\newcommand{\cd}[1]{\textbf{\em\color{green}#1}}
\else  
    \newcommand\cd[1]{}
\fi

\usepackage{balance}

\usepackage[hyphenbreaks]{breakurl}

\begin{document}
\fancyhead{}

\author{Bristena Oprisanu}
\affiliation{%
  \institution{Computer Science Department\\University College London}}
  \email{bristena.oprisanu.10@ucl.ac.uk}
  \author{Christophe Dessimoz}
\affiliation{%
  \institution{University College London and University of Lausanne}}
  \email{c.dessimoz@ucl.ac.uk}

\author{Emiliano De Cristofaro}
\affiliation{%
  \institution{University College London and\\Alan Turing Institute}}
  \email{e.decristofaro@ucl.ac.uk}
\sloppy

  \title{How Much Does GenoGuard Really ``Guard''? An Empirical Analysis of Long-Term Security for Genomic Data\footnotemark}

\begin{abstract}
Due to its hereditary nature, genomic data is not only linked to its owner but to that of close relatives as well. As a result, its sensitivity does not really degrade over time; in fact, the relevance of a genomic sequence is likely to be longer than the security provided by encryption. This prompts the need for specialized techniques providing \emph{long-term} security for genomic data, yet the only available tool for this purpose is GenoGuard~\cite{huang_genoguard:_2015}. By relying on {\em Honey Encryption}, GenoGuard is secure against an adversary that can brute force all possible keys; i.e., whenever an attacker tries to decrypt using an incorrect password, she will obtain an incorrect but plausible looking decoy sequence.

In this paper, we set to analyze the real-world security guarantees provided by GenoGuard; specifically, assess how much more information does access to a ciphertext encrypted using GenoGuard yield, compared to one that was not. Overall, we find that, if the adversary has access to side information in the form of partial information from the target sequence, the use of GenoGuard does appreciably increase her power in determining the rest of the sequence. We show that, in the case of a sequence encrypted using an easily guessable (low-entropy) password, the adversary is able to rule out most decoy sequences, and obtain the target sequence with just 2.5\% of it available as side information. In the case of a harder-to-guess (high-entropy) password, we show that the adversary still obtains, on average, better accuracy in guessing the rest of the target sequences than using state-of-the-art genomic sequence inference methods, obtaining up to 15\% improvement in accuracy.
\end{abstract}

\settopmatter{printacmref=false}

\maketitle

\renewcommand{\thefootnote}{\fnsymbol{footnote}}
\footnotetext[1]{\em To appear in the Proceedings of the 18th ACM CCS Workshop on Privacy in the Electronic Society (WPES 2019).}

\renewcommand*{\thefootnote}{\arabic{footnote}}

\section{Introduction}\label{sec:intro}
Over the past two decades, the cost of sequencing the human genome -- i.e., determining a person's complete DNA sequence -- has plummeted from millions to thousands of dollars, and continues to drop~\cite{genome2017org}. 
As a result, sequencing has not only become routine in biology and biomedics research, but is also increasingly used in clinical contexts, %
with treatments tailored to the patient's genetic makeup~\cite{ashley2016towards}. 
At the same time, the ``direct-to-consumer'' genetic testing market is booming~\cite{adoption} with companies like 23andMe and AncestryDNA attracting millions of customers, and providing them with easy access to reports on their ancestry or genetic predisposition to health-related conditions.
Progress and investments in genomics have also enabled public initiatives to gather genomic data for research purposes. 
For instance, in 2015, the US launched the ``All of Us'' program~\cite{allofus2017}, which aims to sequence one million people, while, in the UK, Genomics England is sequencing the genomes of 100,000 patients with rare diseases or cancer~\cite{genomicsengland}.

Alas, as more and more genomic data is generated, collected, and shared, serious privacy, security, and ethical concerns also become increasingly relevant.
The genome contains very sensitive information related to, e.g., ethnic heritage, disease predispositions, and other phenotypic traits~\cite{ayday2013chills}.
Furthermore, even though most published genomes have been anonymized, previous work has shown that anonymization does not provide an effective safeguard for genomic data~\cite{gymrek_identifying_2013}. 
While some individuals choose to donate their genome to science, or even publicly share it%
~\cite{pgp}, others might be concerned about their privacy, or fear discrimination by employers, government agencies, insurance providers, etc.~\cite{burns_gop_nodate}. %

Worse yet, consequences of genomic data disclosure are not limited in time or to the data owner: due to its hereditary nature, access to one's sequenced genome inherently implies access to many features that are relevant to their progeny and their close relatives. 
A case in point is the story of Henrietta Lacks, a patient who died of cancer in 1951. 
Some of her cancerous cells revealed to be useful for research because of their ability to keep on dividing. 
Unbeknownst to her family, the cells became the most commonly used ``immortal cell line,'' and their genome was eventually sequenced and published~\cite{landry2013genomic}. 
This prompted serious privacy concerns among her family members, even 60 years later~\cite{callaway2013hela}. 

Motivated by these challenges, the research community has produced a large body of work aiming to protect genomic privacy and enable privacy-preserving sharing and testing of human genomes~\cite{sok}.
Available solutions mostly rely on cryptographic tools, including encryption as well as Secure Computation, Homomorphic Encryption, Oblivious RAM, etc.~\cite{aziz2017privacy}.
However, modern encryption algorithms provide security guarantees only against computationally bounded adversary; essentially, their security is assumed to last for 30 to 50 years~\cite{enisa}.
While this timeframe is acceptable for most uses of encryption, it is not for genomic data.

To address the problem of ``long-term security,'' Huang et al.~\cite{huang_genoguard:_2015} introduce GenoGuard, a tool based on Honey Encryption (HE)~\cite{HE2} to provide confidentiality of genomic data even in the presence of an adversary who can brute force all possible encryption keys.
GenoGuard uses a distribution transforming encoder (DTE) together with symmetric (password-based) encryption. In essence, whenever an attacker would try to decrypt a GenoGuard ciphertext using a wrong password, the decryption will give a wrong but plausible looking plaintext, which we denote as a {\em honey sequence.}

HE schemes based on DTE-then-encrypt constructions (as is the case for GenoGuard) only provide security in the message recovery context.
That is, having access to the ciphertext only gives an unbounded adversary a negligible advantage in guessing the correct plaintext.
However, as first discussed by Jaeger et al.~\cite{jaeger2016honey}, ciphertexts obtained from DTE-then-encrypt HE might still leak a significant amount of information about the plaintexts.

\descr{Technical Roadmap.} We evaluate GenoGuard security by analyzing ciphertexts obtained using easily guessable (low-entropy) passwords as well as hard (high-entropy) ones. 
In other words, in both cases, we decrypt a GenoGuard ciphertext using a corpus of passwords and analyze the resulting decryptions (honey sequences).
In the low-entropy setting, we consider an adversary who aims to identify the correct sequence among a pool of honey sequences, whereas, in the high-entropy case, one that uses the GenoGuard ciphertext in order to obtain more information about the target sequence as opposed to inference methods for genomic data.

In our experimental evaluation, we show that, under a low-entropy password setting, an adversary who has access to side information about the target sequence can quickly eliminate the decoy sequences in order to have an increased advantage of guessing the correct sequence. This draws attention to the fact that if the attacker obtains a list of known passwords for a user (as passwords are often compromised and/or re-used), together with some side information about the user's sequence, she can have a significant advantage in guessing the correct sequence.

In the high-entropy setting, not only we observe that access to the GenoGuard ciphertext improves an adversary's accuracy in guessing SNVs from a target sequence when 10\% or less of the target sequence is available to her as side information, but also draws attention to the fact that with enough side information, %
the adversary can predict a significant part of the target genome just by using state of the art inference methods for genomic sequences.

\descr{Contributions.} 
In summary, our paper makes two main contributions.
First, under a low-entropy password setting, we formally show that, if the adversary obtains side information about the target sequence, there is a significant lower bound in her advantage. This highlights that the system offers low security when the adversary has access to side information, as supported by empirical evidence.
Second, in the high-entropy password setting, we quantify the privacy loss for a user as a result of using GenoGuard, compared to the best inference methods for genomic data; once again, showing that that it is non-negligible.

\descr{Paper Organization.} The rest of the paper is organized as follows. The next section reviews notions used throughout the paper, then, in Section~\ref{sec:genoguard}, we introduce GenoGuard. 
Section~\ref{sec:evaluation} presents our evaluation methodology for low and high-entropy settings, %
while Section~\ref{sec:results} reports our experimental results.
Finally, after reviewing related work in Section~\ref{sec:related work}, the paper concludes in Section~\ref{sec:conc}.

\section{Preliminaries}\label{sec:preliminaries}
This section provides some relevant background information used throughout the paper.

\subsection{Genomics Primer}\label{sec:genomics}

\descr{Genome.} In the nucleus of an organism's cell, double stranded deoxyribonucleic acid (DNA) molecules are packaged into thread-like structures called chromosomes. 
DNA molecules consist of two long and complementary polymer chains of four units called nucleotides, described with the letters A, C, G, and T. 
All chromosomes together make up the {\em genome}, which represents the entirety of the organism's hereditary information; in humans, the genome includes 3.2 billion nucleotides.
A {\em gene} is a particular region of the genome that contain the information to produce functional molecules, in particular proteins.
For instance, the BRCA2~\cite{yoshida_role_2004} is a human tumor suppressor gene (it encodes a protein responsible for repairing the DNA),
and a mutation in that gene increases significantly the risk for breast cancer~\cite{friedenson2007brca1}.
{\em Alleles} are the different versions of genes, as organisms inherit two alleles for each gene, one from each parent. 
The set of genes is also called the {\em genotype.} 
Finally, the {\em haplotype} is a group of alleles in an organisms that are inherited together from a single parent \cite{clarke_disentangling_2005}.

\descr{SNPs and SNVs.} Humans share about 99.5\% of the genome, while the rest differs due to genetic variations. 
The most common type of variants are Single Nucleotide Polymorphisms (SNPs)~\cite{reference_what_nodate}, which occur at a single position and in at least 1\% of the population. 
More generally, variants at specific positions of a genome are referred to as Single-Nucleotide Variants (SNVs); they may be due to SNPs, to rare variants in the population, or to new mutations. 
Typically, SNPs and SNVs are encoded with a value in $\{0,1,2\}$, with $0$ denoting the most common variant (allele) in the population, and $1$ and $2$ denoting alternative alleles.

\descr{Allele Frequency (AF).} The frequency of an allele at a certain position in a given population is known as Allele Frequency (AF).
More specifically, it is the ratio of the number of times the allele appears in the population over the total number of copies of the gene.
In a nutshell, it shows the genetic diversity of a species' population. 

\descr{Linkage Disequilibrium (LD).} LD refers to the non-random association of alleles at two or more positions in the general population, defined as the difference between the frequency of a particular combination of alleles at different positions and the one expected by random association. 

\descr{Recombination Rate (RR).} The process of determining the frequency with which characteristics are inherited together is known as recombination. 
This is due to two chromosomes of similar composition coming together and performing a molecular crossover, thus, exchanging the genetic content. 
Because recombination can occur with small probability at any location along the chromosome, the frequency of recombination between two locations depends on the distance separating them. 
Therefore, for genes sufficiently distant on the same chromosome, the amount of crossover is high enough to destroy the correlation between alleles~\cite{li_modeling_2003}. 
The recombination rate (RR), as defined in \cite{philips__nodate}, is the probability that a transmitted haplotype constitutes a new combination of alleles different from that of either parental haplotype. An example of how a haplotype is created by copying parts from the other haplotypes is illustrated in Figure~\ref{fig:recomb}. 
 
\begin{figure}[t]
    \centering
    \includegraphics[width=0.735\columnwidth]{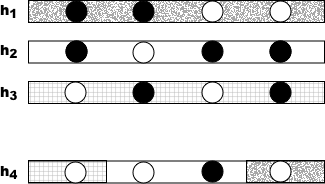}
    \caption{An example of a haplotype, $h_4$, built as an imperfect mosaic from $h_1, h_2, h_3$. $h_4$ is created by (imperfectly) ``copying'' parts from $h_1, h_2$, and $h_3$. Each column of circles represents a SNP locus, with the black and white circles denoting the two alleles -- major and minor. (Adapted from~\cite{li_modeling_2003}).}
    \label{fig:recomb}
\end{figure} 
\subsection{Markov Chains}

A Markov chain is a probabilistic model encoding a sequence of possible events: the probability of each one of them depends only on the state attained in the previous event~\cite{MarCha}.

In the context of genomes, a Markov chain can represent a series of SNVs ordered by their positions. 
In particular, a $k$-th order Markov chain, on genome sequences, can be used to encode a set of SNVs, where the value of each SNV$_i$ depends on the values of the $k$ preceding ones:\vspace*{-0.1cm}
\begin{equation}
\Pr(SNV_i) = \Pr(SNV_i|SNV_{i-1},\hdots,SNV_{i-k})
\end{equation}

\subsection{SNV Correlation Modeling}\label{sec:correlation}
In order to model correlations between SNVs, and perform sequence inference (i.e. predicting the values of SNVs from a sequence), one can use a few different approaches 
(for more details on various SNV correlations, please refer to~\cite{samani_quantifying_2015}).
We choose three models; see next.

\descr{Most likely genotype.} First, we use a model based on the 1st order Markov chain model from AF and LD. Given allele frequencies (AF) and linkage disequilibrium (LD), we predict each SNV using the highest conditional probability of the SNV occurring.
For each SNV, the joint probability matrix is computed taking into consideration the LD with previous one and the AF. 
If a SNV is not in LD with the previous one, the probability is computed using only the allele frequency. 
When this model is used for inference, the highest value from the joint probability matrix or the highest probability given by the AF is chosen to predict the specific SNV.

\descr{Sampled genotype.} The second model is built from the 1st order Markov chain model from AF and LD.
For this model, the conditional probabilities are computed in a similar way as in the most likely genotype model.
The main difference is in the choice of the value of the SNV, given the three computed probabilities for major homozygous $\Pr_0$, heterozygous $\Pr_1$, and minor homozygous allele $\Pr_2$. 
A seed $s$ is chosen uniformly at random from the interval $[0,1)$. If $s<\Pr_0$, then choose the SNV to be major homozygous; if $\Pr_0 \leq s<\Pr_1 +\Pr_0$, then the SNV is heterozygous; and minor homozygous otherwise.

\descr{RR Model.} This is a high-order correlation model that relates LD patterns to the underlying recombination rate~\cite{Li2213}. 
Given a set of $n$ sampled haplotypes, $\{h_1, h_2,...,h_n\}$, the model relates their distribution to the underlying recombination rate. 
Given the recombination parameter, $\rho$, we have:
\begin{equation}
\begin{split}
&\Pr(h_1,...,h_n|\rho) =\\
&=\Pr(h_1|\rho)\cdot\Pr(h_2|h_1;\rho)\cdot\ldots\cdot\Pr(h_n|h_1,\hdots ,h_{n-1};\rho)
\end{split}
\end{equation}
We use this model to determine the value of a SNP at a given position. 
At each SNP, $h_k$ is a possibly imperfect copy of one of $h_1,...,h_{k-1}$. Let $H_i$ denote which haplotype is copied at a position $i$. 
For instance, in the example presented in Figure~\ref{fig:recomb}, for $h_4$, we have $(H_1,H_2,H_3, H_4) = (3,2,2,1)$. 
For a generic $h_k$, each $H_i$ can be modeled as a Markov chain on $\{1,\hdots, k-1\}$. 
Assuming that one part of $h_{k}$ comes from $h_i$, the next adjacent part can be copied from any of the $k-1$ haplotypes, and the probability depends on the recombination rates between these two parts. 
Overall, the probability of a particular haploid genotype $h_{k}$ can be computed as the sum over all possible event sequences of recombination and mutation that could lead to $h_{k}$. 
Let $h_{i,j+1}$ denote the allele found at position $j+1$ in haplotype $i$, and $h_{i,\leq j}$ denote the values of the first $j$ positions of haplotype $i$ (i.e. the prefix sequence of $h_{i,j+1}$). 
Then, we can compute the conditional probability of an allele $h_{k,j+1}$, given all preceding alleles as:
\begin{equation}
\Pr(h_{k,j+1}| h_{k,j}, \hdots h_{k,1}) = \frac{\Pr(h_{k,\leq j+1})}{\Pr(h_{k,\leq j)}}
\end{equation}

\subsection{Honey Encryption}\label{sec:he}
Honey Encryption (HE)~\cite{HE2} %
is a cryptographic primitive used to provide confidentiality guarantees in the presence of possible brute-force attacks.
It is a variant of Password-based Encryption (PBE), in that it also uses an arbitrary string (password) to perform randomized encryption of a plaintext. 
Its main property is that all decryptions of a ciphertext will yield a plausible-looking plaintext, which is thus indistinguishable from the correct one. 

The main building block of HE is the Distribution-Transforming Encoder (DTE).
A DTE is a randomized encoding scheme {\fontfamily{cmtt}\selectfont(encode,\\ decode)} tailored on the target distribution.
The {\fontfamily{cmtt}\selectfont encode} algorithm takes as input a message $M$ from the message space $\mathcal{M}$, and outputs a value $S$ in a set $\mathcal{S}$, i.e., the seed space. 
Whereas, {\fontfamily{cmtt}\selectfont decode} takes a seed $S \in \mathcal{S}$ and outputs a message $M \in \mathcal{M}$. 
A DTE scheme is {\em correct} if, for any $M \in \mathcal{M}$, $\Pr[${\fontfamily{cmtt}\selectfont decode$($encode$(M))$}$=M]=1$. 
The DTE-then-encrypt scheme presented in~\cite{HE2} applies {\fontfamily{cmtt}\selectfont encode} to a message, and then performs encryption using a secure symmetric encryption scheme (e.g., AES).
Similarly, to decrypt a ciphertext, one first decrypts using the underlying cipher (e.g., AES), and then applies the {\fontfamily{cmtt}\selectfont decode} algorithm.

\descr{Terminology.} In the rest of the paper, to denote sequences decrypted from GenoGuard, we use the term \emph{honey sequences}.

\begin{figure*}[t]
    \centering
    \includegraphics[width=0.9\textwidth]{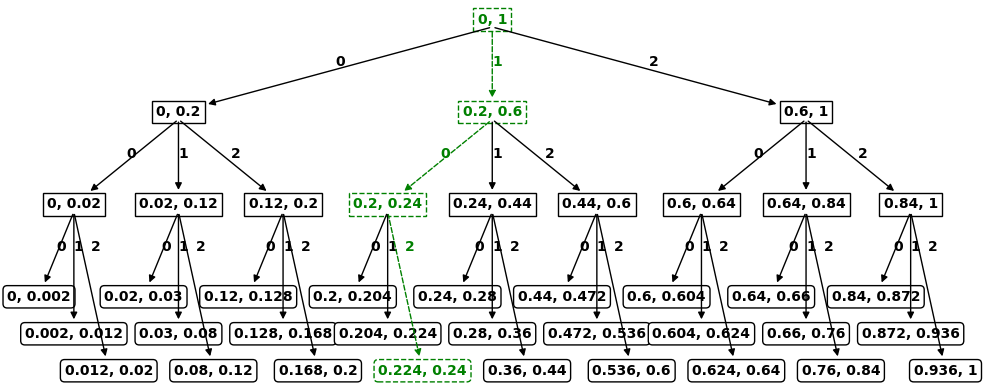}
    \vspace{-0.2cm}
    \caption{Toy example describing the encoding process for a sequence $(1, 0, 2)$. The green dashed line represents the correct encoding of the sequence. When the final leaf (interval $[0.224, 0.24]$) is reached, a seed is picked at random from this range. }
    \label{fig:tree}
\end{figure*} 

\section{GenoGuard}\label{sec:genoguard}
In this section, we review GenoGuard~\cite{huang_genoguard:_2015}, along with a security analysis of the framework.%

\subsection{Construction}\label{subsec:ggdesc}
GenoGuard is a framework %
providing long-term confidentiality for genomic data based on Honey Encryption~\cite{HE2}.
More specifically, it  allows to encode genomic data, encrypt it using a secret password, and store in a database, in such a way that its confidentiality is preserved even against an attacker that can brute-force all possible passwords.
In GenoGuard, genomes are represented as a sequence of single-nucleotide variants (SNVs), i.e., values in $\{0,1,2\}$.

\descr{Encoding.} The construction uses a DTE scheme optimized for genome sequences. 
It assigns subspaces of seed space $\mathcal{S}$ to the prefixes of a sequence $M$, i.e., all the subsequences in the set $\{M_{1,i}| 1\leq i \leq n\}$, where $n$ is the length of the sequence. 
For example, the prefixes of the sequence $01102$ are $\{0, 01, 011, 0110, 01102\}$. 
The seed space $\mathcal{S}$ is the interval $[0,1)$, with each seed being a real number in this interval. 

Let $\mathcal{M}$ be the set of all possible sequences (the plaintext space).
To calculate the cumulative distribution function (CDF) of each sequence, a total order $\mathcal{O}$ is assigned to all sequences in $\mathcal{M}$. 
For any two different sequences $M$ and $M'$, we assume that they start to differ at SNV$_i$ and SNV$'_i$. %
If the value of SNV$_i$ is smaller than that of SNV'$_i$, then, $\mathcal{O}(M)< \mathcal{O}(M')$, and $\mathcal{O}(M)> \mathcal{O}(M')$ otherwise. 
The CDF of a sequence $M$ is then calculated as:
\begin{center}
$CDF(M)=\sum_{M' \in \mathcal{M},\mathcal{O}(M')\leq \mathcal{O}(M)}  \Pr_{SNV}(M')$
\end{center}
where $\Pr_{SNV}(M')$ is the probability of the sequence $M'$.

The encoding of a sequence can be performed using a perfect ternary tree, as depicted in Figure~\ref{fig:tree}. 
(Note that the plot was generated using code obtained from GenoGuard's Github page.\footnote{\url{https://github.com/acs6610987/GenoGuard}})
Each node in the tree represents a prefix of a sequence, and each leaf a complete sequence. 
Nodes have an interval $[L_i^j, U_i^j)$, where $i$ is the depth of the node in the tree and $j$ its order at a given depth $i$. 
The first node has the interval $[L_0^0, U_0^0) = [0,1)$. 
Depending on the value of the SNV at position $i{+}1$, the encoding proceeds from the node that represents $M_{1,i}$ with order $j$ at depth $i$ to depth $i+1$ as follows:
\begin{itemize}
\item[$\bullet$] If SNV$_{i+1} = 0$, go to the left branch and attach an interval $[L_{i+1}^{3j},U_{i+1}^{3j}) = [L_i^j, L_i^j + (U_i^j-L_i^j)\times \Pr(SNV_{i+1}=0|M_{1,i}))$
\item[$\bullet$] If SNV$_{i+1}=1$, go to the middle branch and attach an interval $[L_{i+1}^{3j+1},U_{i+1}^{3j+1}) = [L_i^j + (U_i^j-L_i^j)\times \Pr(SNV_{i+1}=0|M_{1,i}), L_i^j + (U_i^j-L_i^j)\times( \Pr(SNV_{i+1}=0|M_{1,i} + \Pr(SNV_{i+1}=1|M_{1,i})))$
\item[$\bullet$] If SNV$_{i+1} = 2$, go to the right branch and attach an interval $[L_{i+1}^{3j+2},U_{i+1}^{3j+2}) = [L_i^j + (U_i^j-L_i^j)\times( \Pr(SNV_{i+1}=0|M_{1,i} + \Pr(SNV_{i+1}=1|M_{1,i})), U^j_i)$.
\end{itemize}
In order to compute the conditional probabilities, Huang et al.~\cite{huang_genoguard:_2015} consider several models and compare their goodness of fit for real-world genome datasets.
Specifically, they experiment with Linkage Disequilibrium (LD), Allele Frequencies (AF), building $k$-th order Markov chains on the dataset and recombination rates (RR), and find the latter to perform best.

Finally, when a leaf is reached, a seed is picked uniformly from this range as the encoding of the corresponding sequence, and then fed into a Password-based Encryption (PBE) scheme to perform encryption, using a password chosen by the user. %

\descr{Decoding.} To decode an encoded-then-encrypted sequence, the ciphertext is first decrypted (as per the PBE scheme) using the user-chosen password; this recovers the seed.
Then, the decoding process proceeds similar to the encoding one. 
That is, given the seed $S \in [0,1)$, at each step, the algorithm computes three intervals for the three branches, chooses the interval in which the seed $S$ falls, and moves down the tree. 
Once a leaf node is reached, the path from the root to the leaf is outputted as the decoded sequence.

\descr{Finite Precision.} Note that the Honey Encryption encoding model, as described in Section~\ref{subsec:ggdesc}, requires the seed space $\mathcal{S}$ to be a real number domain with infinite precision. 
In the case of DNA sequences, this would yield a very long floating-point representation, and thus a high storage overhead. 
Therefore, GenoGuard uses a modification of the DTE scheme for finite precision. %
Specifically, for a sequence of length $n$, where each SNV takes three possible values, at least $n \cdot \log_23$ bits are needed for storing the sequence.
Hence, a storage overhead parameter $h > \log_23$ is selected, and each sequence is encoded over $h\cdot n$ bits. 
The algorithm works as before, by selecting intervals according to the values of the respective SNVs based on conditional probabilities. The root interval is $[0,2^{hn}{-}1]$. 
At each branch at depth $i$, the algorithm will allocate a seed space of size $3^{n-i-1}$, and each following step will segment an input interval into three parts of equal size. 
Hence, any subinterval of the $j$-th node at depth $i$ will contain $3^{n-i-1}$ integers.

\subsection{Security}\label{sec:security}

Huang et al.~\cite{huang_genoguard:_2015} evaluate the security of GenoGuard vis-\`a-vis the probability of an unbounded adversary recovering the encrypted sequence. 
That is, given the encryption of a message, what is the probability of the adversary recovering the correct message, even if she can brute-force all possible encryption keys for the underlying PBE scheme?

\descr{Upper Bound.} More formally, they prove an upper bound to the probability an adversary recovers the correct message to be:
\begin{equation}\label{eq:secu}
Pr_{p_m,p_k} \leq w(1+\delta) + \frac{3^n +1/w}{2^{(h-\log_23)n}}
\end{equation}
where  $p_m$ is the original sequence distribution with maximum sequence probability $\gamma$, $p_k$ is a key (password) distribution with maximum weight $w$ (i.e., the most probable password has probability $w$), $n$ is the length of the sequence, $h$ the overhead parameter, and $\delta$ a parameter depending on $w$ and $\gamma$.

Let $\Delta$ denote the fraction $\frac{3^n +1/w}{2^{(h-\log_23)n}}$ in Equation~\ref{eq:secu}.
Note that $\Delta$ is a security loss term, since the upper bound on plaintext recovery probability should be $w$, as an adversary who trivially decrypts the ciphertext with the most probable key and outputs the result can recover the original message with probability $w$.
$\Delta$ is essentially the security lost due to DTE imperfectness when moving to finite precision, i.e., given by the difference between the original message distribution and the DTE distribution.
As shown in~\cite{huang_genoguard:_2015}, for $n = 20{,}000$, $h=4$, $w = \frac{1}{100}$, and $\gamma = 2.89\times 10^{-44}$, $\Delta$ is approximately $2^{-16600}$.

\descr{Empirical Evaluation.} Huang et al.~\cite{huang_genoguard:_2015} also present an {\em empirical} security analysis based on two experiments.
In both, the chromosome 22 of a victim is encrypted using a password pool consisting of numbers from 1 to 1000, with ``539'' assumed to be the correct one. 
Then, in order to rule out wrong passwords, the interval size of each of the decrypted sequences is computed.
In the first experiment, a genome is encoded by assuming a uniform distribution (i.e., each branch has weight $1/3$ at all depths), and a PBE scheme is used to encrypt the seed. %
In the second experiment, GenoGuard is used to encrypt the victim's sequence.
Hence, the size of the interval of a leaf in the ternary tree is proportional to the probability of the corresponding sequence. 
The results of their experiments, reported in Figure 10 in~\cite{huang_genoguard:_2015}, show that a simple classifier can distinguish the correct sequence in the first experiment, while, in the second one, it is ``buried'' among all the decrypted sequences. 

\section{Evaluation Methods}\label{sec:evaluation} 
We now describe our evaluation methods, for both low and high-entropy password settings. 
Before doing so, we introduce the notation used in the rest of the paper in Table~\ref{tab:notations}.

\subsection{Low-Entropy vs High-Entropy Password}
We use different approaches for evaluating GenoGuard under two different password types, namely low-entropy and high-entropy passwords. In other words, we encrypt a sequence with GenoGuard using either an easy to guess, low-entropy password ($\approx$7 bits), or using a harder password with a higher entropy ($\approx$72 bits). 

The difference in the evaluation of the two approaches is given by the adversary's goal. 
Specifically, in the low-entropy password case, the adversary attempts to use the side information in order to distinguish the original encrypted sequence among a pool of honey sequences.
By contrast, in the high-entropy setting, the adversary uses both the honey sequences and the side information in order to predict the value of each SNV at each position in the target sequence.

\begin{table}
\small
\begin{center}
\begin{tabular}{l|l}
{\bf Symbol} & {\bf Meaning}\\ \hline 
MR & Message recovery\\ %
SI & Side information \\ %
HEnc & Honey Encryption \\ %
HDec & Honey Decryption \\ %
$\mathcal{K}$ & Key space\\ %
$\mathcal{M}$ & Message space\\ %
$p_k$ & Key distribution\\ %
$p_m$ & Message distribution \\ %
$\adv$ & Adversary \\ %
$\adv^{SI}$ & Adversary with access to side information \\ 
\end{tabular}
\caption{Notation.}
\label{tab:notations}
\end{center}
\end{table}

\subsection{Threat Model}\label{subsec:threat-low}
We use the same system and threat model presented in the GenoGuard paper~\cite{huang_genoguard:_2015},
i.e., we assume a genomic sequence of a user is to be stored, encrypted, at a third-party database, e.g., a biobank.
We consider an adversary that has access to the encrypted data (for instance, she breaks into the biobank and gets access to the encrypted database, or the biobank itself is adversarial) and has access to public knowledge as well as to some side information (as discussed below). 

\descr{Low-Entropy Password.}
The main adversarial goal in this case is to identify the target sequence among a pool of honey sequences, using the side information available, i.e. ``message recovery'' with side information (\textbf{MR-SI}). 

\descr{High-Entropy Password.} The main adversarial goal is to obtain as much information as possible about the sequence that was encrypted.
Note that this adversarial goal is different from ``message recovery,'' according to which Huang et al.~\cite{huang_genoguard:_2015} evaluate GenoGuard's security (cf.~Section~\ref{sec:security}).
The main intuition is that, as also hypothesized by \cite{jaeger2016honey}, using Honey Encryption might actually leak non-negligible information about the sequences encrypted using GenoGuard, even if the adversary cannot correctly recover the full plaintext with non-negligible probability.

\subsection{Adversary's Side Information}\label{subsec:adv}
As mentioned above, the adversary has access to the victim's encrypted sequence as well as to public information such as, Linkage Disequilibrium, Allele Frequencies, Recombination Rate (see Section~\ref{sec:genomics}).
In addition, we assume that the adversary may have some side information about the victim.

When referring to side information, note that we do {\em not} consider knowledge of common traits from phenotype-genotype associations, e.g., gender, ancestry, or other information about the victim that could be obtained, e.g., from social media.
In fact, this is covered by GenoGuard's guidelines, which state that the user should include as much side information about their genome as possible when performing the encoding.
Whereas, even though assuming the user can knowingly enumerate all possible side information is quite a strong assumption, we actually consider the case where the victim undertakes some specific tests, and the adversary learns additional information about the victim from the outcome of those tests.
Additionally, the victim might choose to re-encrypt their genome after obtaining the test results in order to incorporate them in the encoding, and the adversary could use the new ciphertext to extract information about the old ciphertext.

In the high-entropy password setting, we also evaluate the case where an adversary has no side information about the target sequence, in order to quantify the information leakage that might occur from using GenoGuard against baseline inference methods for genomic sequences.
Overall, we consider different types of side information available to the adversary:
\begin{compactenum}
\item \emph{No Side Information:} The adversary has access only to the encrypted sequence. (NB: this is only evaluated for the high-entropy password setting)
\item \emph{Sparse SNVs:} The adversary has access to SNV values sparsely distributed in the target sequence. 
\item \emph{Consecutive SNVs:} The adversary has access to values from a cluster of consecutive SNVs in the target sequence. 

\end{compactenum}
\subsection{Low-Entropy Password} \label{sec:low ent}

We now formally provide a lower bound for the adversary's advantage in the case where she obtains side information about the target sequence and encryption is done using a low-entropy password.

We present a lower bound on the adversary's advantage when she has access to side information about the encrypted sequence and can exhaustively search the message space.
We prove the bound formally, building on~\cite{jaeger2016honey}, which shows the impossibility of known-message attack (KMA) security with low-entropy passwords.
However, instead of the adversary having access to message-ciphertext pairs, we assume that the adversary has access to (position, value) pairs from the encrypted sequence. 
The game defining message recovery security with side information is denoted as \textbf{MR-SI$^\adv_{HE,p_m, p_k}$} and illustrated in Figure~\ref{fig:mrsi}.  

Given a ciphertext $C^*$, an adversary $\adv^{SI}$, with access to side information, is allowed to guess the message by brute force. The adversary $\adv^{SI}$ wins the game if her output message is the same as the original message.

\begin{figure}[tbp]
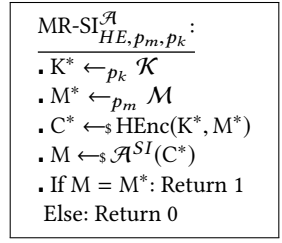

\center
\framebox{
\begin{tabular}[t]{l}
\underline{MR-SI$^\adv_{HE,p_m, p_k}$:}\\
$\centerdot$ $\mathrm{K}^* \gets_{p_k} \mathcal{K}$\\
$\centerdot$ $\mathrm{M}^* \gets_{p_m} \mathcal{M}$\\
$\centerdot$ $\mathrm{C}^* \sample \mathrm{HEnc}(\mathrm{K}^*, \mathrm{M}^*)$\\
$\centerdot$ $\mathrm{M} \sample \adv^{SI}(\mathrm{C}^*)$\\
$\centerdot$ If $\mathrm{M = M}^*$: Return 1\\
~~Else: Return 0\\
\end{tabular}
}
\vspace*{-0.2cm}
\caption{Definition of Message Recovery Security with Side Information (MR-SI).}\label{fig:mrsi}
\vspace*{-0.2cm}
\end{figure}

\begin{figure}[tp]
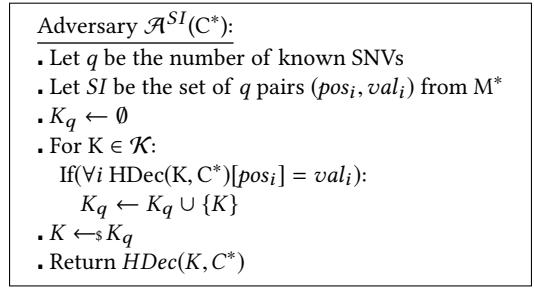

\center
\framebox{
\begin{tabular}[t]{l}
\underline{Adversary $\adv^{SI}$(C$^*$):}\\
$\centerdot$ Let $q$ be the $ \mathrm{number\ of\ known\ SNVs}$\\
$\centerdot$ Let $SI$ be the $\mathrm{ set\ of\ }q \mathrm{\ pairs\ } (pos_i, val_i) \mathrm{\ from\ M}^*$\\
$\centerdot$ $K_q \leftarrow \emptyset $\\
$\centerdot$ For $ \mathrm{K} \in \mathcal{K}$:\\
 \hspace{3mm}If$(\forall i \ \mathrm{HDec(K, C^*})[pos_i] = val_i)$:\\
\hspace{5mm} $K_q \leftarrow K_q \cup \{K\}$\\
$\centerdot$ $K \sample K_q $\\
$\centerdot$ Return $HDec(K, C^*)$\\
\end{tabular}
}
\vspace*{-0.2cm}
\caption{Adversary strategy for MR-SI, having access to $q$ pairs of (position, value) from the original message.}\label{fig:a-si}
\end{figure}

Our intuition is that the advantage of the adversary \adv$^{SI}$ (Figure~\ref{fig:a-si}), for a number $q$ ($q\leq 2n$, where  $n= [\log_2|\mathcal{K}|]$) of positions and values, from the original sequence, is equal to the probability that a randomly chosen key that decrypts correctly all values at the given positions, will also decrypt the rest of the sequence, i.e., \advantage{MR-SI}{HE,p_m,p_k} = $\Pr[\mathrm{MR-SI}_{HE,p_m,p_k}^\adv]$. We denote by $K_q$ the number of keys consistent with the positions and values used as side information. 

Hence, we use Lemma 4.2 from \cite{jaeger2016honey}, as follows:
\begin{lemma} \label{lemma}
If $s_0, s_1, ..., s_{n}$ are positive integer-valued random variables such that $s_0\leq2^{n}$ and $s_{q+1}\leq s_q$, for $q\in \mathbb{Z}_{n}$, then $\mathrm{max}_{q\in \mathbb{Z}_{n}} \expect{s_{q+1}/s_q}\geq \frac{1}{2n}$.
\end{lemma}
\begin{proof}
See \cite{jaeger2016honey}.
\end{proof}

Using Lemma~\ref{lemma}, we can compute the adversary's advantage as follows:
\begin{theorem}
Let HE be an encryption scheme and $n = [\log_2|K|]$. Then, for any $p_m, p_k$, the adversary \adv$^{SI}$ who obtains at most $n-1$ positions and values from the original sequence will have advantage:
\begin{center}
\advantage{MR-SI}{HE,p_m,p_k}[(\adv^{SI})] $\geq \frac{1}{2n^2}$
\end{center} \label{thm}
\end{theorem}

\begin{figure}[tp]
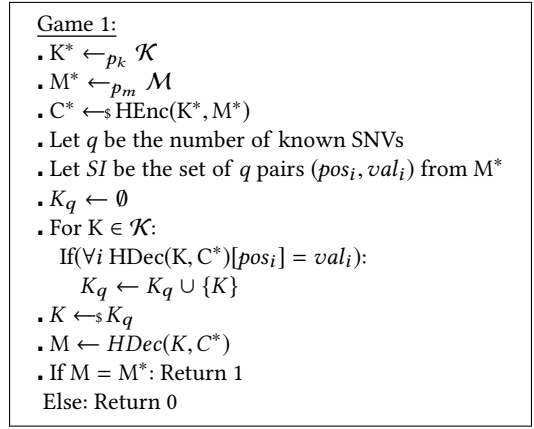

\center
\framebox{
\begin{tabular}[t]{l}
\underline{Game 1:}\\
$\centerdot$ $\mathrm{K}^* \gets_{p_k} \mathcal{K}$\\
$\centerdot$ $\mathrm{M}^* \gets_{p_m} \mathcal{M}$\\
$\centerdot$ $\mathrm{C}^* \sample \mathrm{HEnc}(\mathrm{K}^*, \mathrm{M}^*)$\\
$\centerdot$ Let $q$ be the $ \mathrm{number\ of\ known\ SNVs}$\\
$\centerdot$ Let $SI$ be the $\mathrm{ set\ of\ }q \mathrm{\ pairs\ } (pos_i, val_i) \mathrm{\ from\ M}^*$\\
$\centerdot$ $K_q \leftarrow \emptyset $\\
$\centerdot$ For $ \mathrm{K} \in \mathcal{K}$:\\
 \hspace{3mm}If$(\forall i \ \mathrm{HDec(K, C^*})[pos_i] = val_i)$:\\
\hspace{5mm} $K_q \leftarrow K_q \cup \{K\}$\\
$\centerdot$ $K \sample K_q $\\
$\centerdot$ $\mathrm{M} \leftarrow HDec(K, C^*)$\\
$\centerdot$ If $\mathrm{M = M}^*$: Return 1\\
~~Else: Return 0\\
\end{tabular}
}
\vspace*{-0.2cm}
\caption{Game 1, used in the proof of Theorem \ref{thm}. }\label{fig:g1}
\end{figure}

\begin{proof}
The advantage $\advantage{MR-SI}{HE,p_m,p_k}[(\adv^{SI})]$, is equal to $\Pr[\mathrm{Game\ 1\ Retuns\ 1}]$ where Game 1 is defined in Figure~\ref{fig:g1}. This is due to the fact that Game 1 is MR-SI$^\adv_{HE,p_m, p_k}$ together with Adversary $\adv^{SI}$(C$^*$). 
By applying a few transformations to Game 1 and changing the final check, i.e. instead of checking if $M=M^*$ before returning 0 or 1, it checks if the key $K$ is in the subset that decrypts $C^*$ to $M^*$ we obtain an equivalent game, Game 2 (Figure~\ref{fig:g2}). Thus, $\Pr[\mathrm{Game\ 1}  \mathrm{\ Returns\ } 1] = \Pr[\mathrm{Game\ 2}  \mathrm{\ Returns\ } 1]$.

Since $K_{q+1}\subseteq K_q$, for fixed $q$, the probability that Game 2 will return 1 is $\expect{\frac{|K_{q+1}|}{|K_q|}}$. So we have $\Pr[\mathrm{Game\ 2}  \mathrm{\ Returns\ } 1]  = \sum^{n}_{q=0}\frac{1}{n}\expect{\frac{|K_{q+1}|}{|K_q|}}\vspace{0.1cm}$.

We then define Experiment 1 (Figure~\ref{fig:e1}), which shows that the distribution of $K_{q+1}$ and $K_q$ for $q \in \mathbb{Z}_n$ is the same as the distribution in Game 1. Let $s_q$ denote $|K_q|$ and $\epsilon = max_{q\in \mathbb{Z}_n}\expect{\frac{s_{q+1}}{s_q}}$, where the expectation is taken in Experiment 1. Since all $K_q$ contain at least the key $K^*$, they all are positive. Thus, by applying Lemma~\ref{lemma} we have $\epsilon \geq \frac{1}{2n}$. Then:
\begin{center}
\advantage{MR-SI}{HE,p_m,p_k}[(\adv^{SI})]  = $\Pr[\mathrm{Game\ 2\ Returns\ 1} ]\newline  =  \sum^{n}_{q=0}\frac{1}{n}\expect{\frac{|K_{q+1}|}{|K_q|}} \geq \frac{1}{n}\cdot \epsilon \geq \frac{1}{2n^2}$
\end{center}
\vspace{-0.2cm}
\end{proof}

This shows that the security of the systems is weak even with a small number of pairs (position, value) from the target sequence available to the attacker, as opposed to having multiple known ciphertext-plaintext pairs.

 \begin{figure}[tp]
\vspace*{-0.3cm}
\center
\framebox{
\begin{tabular}[t]{l}
\underline{Game 2:}\\
$\centerdot$ $\mathrm{K}^* \gets_{p_k} \mathcal{K}$\\
$\centerdot$ $\mathrm{M}^* \gets_{p_m} \mathcal{M}$\\
$\centerdot$ $\mathrm{C}^* \sample \mathrm{HEnc}(\mathrm{K}^*, \mathrm{M}^*)$\\
$\centerdot$ Let $q$ be the $ \mathrm{number\ of\ known\ SNVs}$\\
$\centerdot$ Let $SI$ be the $\mathrm{ set\ of\ }q \mathrm{\ pairs\ } (pos_i, val_i) \mathrm{\ from\ M}^*$\\
$\centerdot$ $K_0 \leftarrow \mathcal{K}; K_1, K_2, ...K_{q+1} \leftarrow \emptyset $\\
$\centerdot$ For $  (pos_i, val_i) \in SI$:\\
\hspace{3mm}For $ \mathrm{K} \in K_{i-1}$:\\
\hspace{5mm}If$(\mathrm{HDec(K, C^*})[pos_i] = val_i)$:\\
\hspace{7mm} $K_i\leftarrow K_i \cup \{K\}$\\
$\centerdot$ For $K \in K_q $:\\
 \hspace{3mm} If $HDec(K, C^*) = M^*$\\
 \hspace{5mm} $K_{q+1} \leftarrow K_q \cup \{K\}$\\
$\centerdot$ $K \sample K_q $\\
$\centerdot$ If $K \in K_{q+1}$: Return 1\\
~~Else: Return 0\\
\end{tabular}
}
\vspace*{-0.2cm}
\caption{Game 2, a transformed version of Game 1.}\label{fig:g2}
\end{figure}

 \begin{figure}[tp]
\center
\framebox{
\begin{tabular}[t]{l}
\underline{Experiment 1:}\\
$\centerdot$ $K_0 \leftarrow \mathcal{K}; K_1, K_2, ...K_{n} \leftarrow \emptyset $\\
$\centerdot$ $\mathrm{K}^* \gets_{p_k} \mathcal{K}$\\
$\centerdot$ $\mathrm{M}^* \gets_{p_m} \mathcal{M}$\\
$\centerdot$ $\mathrm{C}^* \sample \mathrm{HEnc}(\mathrm{K}^*, \mathrm{M}^*)$\\
$\centerdot$ Let $n$ be the $ \mathrm{number\ of\ known\ SNVs }$\\
$\centerdot$ Let $SI$ be the $\mathrm{ set\ of\ }n \mathrm{\ pairs\ } (pos_i, val_i) \mathrm{\ from\ M}^*$\\
$\centerdot$ For $  (pos_i, val_i) \in SI$:\\
 \hspace{3mm}For $ \mathrm{K} \in K_{i-1}$:\\
\hspace{5mm}If$(\mathrm{HDec(K, C^*})[pos_i] = val_i)$:\\
\hspace{7mm} $K_i\leftarrow K_i \cup \{K\}$\\

\end{tabular}
}
\vspace*{-0.2cm}
\caption{Experiment 1, used in the proof of Theorem~\ref{thm}.}\label{fig:e1}
\end{figure}

\subsection{High-Entropy Password}\label{sec:attacks-high} 

We now give an overview of our inference strategy using the GenoGuard ciphertext and discuss the baseline inference methods we evaluate our strategy against.

\subsubsection{Baseline Inferences}\label{sec:inference}
We compare the performance of our inference strategy to baselines for genomic sequence inference. %
For these baselines, we assume that the adversary has access only to side information, as discussed in Section~\ref{subsec:adv}, but not the ciphertext resulting from GenoGuard's encode-then-encrypt method.

As done by Samani et al.~\cite{samani_quantifying_2015}, we set to infer the value of an unknown SNV$_i$, given a probabilistic modeling of genome sequences.
More specifically, we use the following models for SNV correlation:
\begin{compactitem}
\item \emph{B1:} 1st order Markov chain model from AF and LD: most likely genotype.
\item \emph{B2:} 1st order Markov chain model from AF and LD: sampled genotype.
\item \emph{B3:} RR Model.
\end{compactitem}
\subsubsection{GenoGuard Inference Methods} \label{subsec:attack-high}

Our method is based on exploiting the similarities between the honey sequences in order to obtain information about the target sequence. More specifically, we use two strategies:
\begin{enumerate}
	\item \emph{G1.} Side information-weighted SNVs: We assign a weight to each of the honey sequences according to the amount of side information contained. We then consider only the sequences with the highest weight and output the most common SNVs among them as our candidate SNVs for the target sequence. In the case of no side information, we consider the most common SNVs across all honey sequences.
	\item \emph{G2.} Interval and Side information-weighted SNVs: Similar to the previous method, however, we also adjust the weight of each sequence when considering the most common SNVs by the size of the interval that the seed of the respective sequence will fall into. In the case of no side information, we take the most common SNVs from all honey sequences, weighted by the previously mentioned interval size.

\end{enumerate}

\begin{figure*}[t]
  \centering
  \begin{subfigure}[b]{1.0\columnwidth}
  \centering
  \includegraphics[width=0.85\columnwidth]{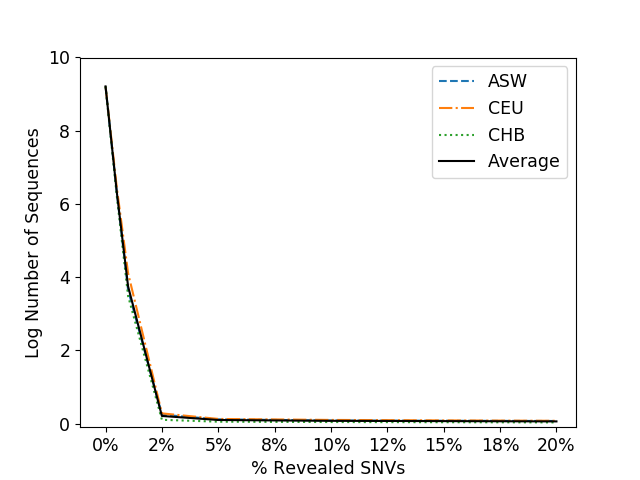}
  \caption{\textmd{\#Candidate sequences vs \% revealed sparse SNVs from target sequence}}
  \label{fig:low_rnd_seqs}
	\end{subfigure}
	\hfill
  \begin{subfigure}[b]{1.0\columnwidth}
  \centering
  \includegraphics[width=0.85\columnwidth]{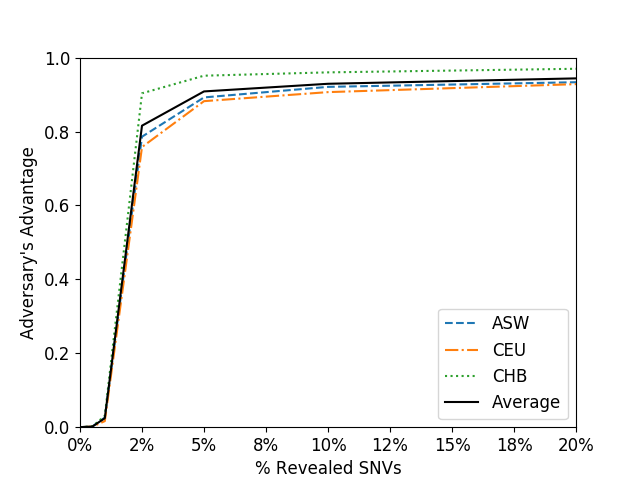}
  \caption{\textmd{Adv's advantage vs \% revealed sparse SNVs from target sequence}}
  \label{fig:low_rnd}
  	\end{subfigure}
  \vspace{-0.2cm}
  \caption{Results of our evaluation in the low-entropy setting vis-\`a-vis an adversary with access to side information in the form of sparse SNVs from the target sequence.}
  \vspace{-0.2cm}
\end{figure*}

\section{Experimental Evaluation} \label{sec:results}

In this section, we present the datasets used for the experimental evaluation and the results obtained for both evaluation methods, %
i.e., low-entropy and high-entropy passwords.

\subsection{Dataset} \label{subsec:data}
We use the Phase III data from the HapMap dataset, i.e., the third release from the HapMap project.\footnote{\url{https://www.sanger.ac.uk/resources/downloads/human/hapmap3.html}}
HapMap was an international project~\cite{international2003international}, run between 2002 and 2009, aimed at developing a haplotype map of the human genome, and describe the common patterns of human genetic variation. 
The HapMap data has been made publicly available and used for various research purposes, e.g., to research genetic variants affecting health, disease and responses to drugs and environmental factors, etc.

The Phase III release increased the number of DNA samples to 1,301 and included 11 different populations.
In our experiments, we select data from three populations:
\begin{compactenum}
	\item ASW (African ancestry in Southwest USA),
	\item CEU (Utah residents with Northern and Western European ancestry from the CEPH collection),
	\item CHB (Han Chinese in Beijing, China).
\end{compactenum}
We sample 50 sequences at random from each of them, for a total of 150 sequences.

For all three populations presented above, we test the same SNVs positions.

\subsection{Low-Entropy Password} \label{subsec:low-results}

\subsubsection{Experiment Overview} \label{subsec:attack-low}
We use the following strategy for our evaluation:
\begin{compactenum}
\item Encrypt a sequence using GenoGuard's DTE-then-encrypt method: for each of the 150 sequences, we select and encrypt 1,000 positions from chromosome 13, with a storage overhead $h = 4$ (the same as in the experimental evaluation of GenoGuard), using a low-entropy password.
\item Decrypt the ciphertext, using the top 10,000 most common passwords released by Daniel Miessler\footnote{\url{https://github.com/danielmiessler/SecLists/blob/master/Passwords/Common-Credentials/10k-most-common.txt}}  (with the encryption password  in the set), to obtain plausible looking honey sequences;
\item Exclude the sequences which do not contain the side information.
\item Output the number of remaining sequences, given how many of the possible passwords match the side information.
\end{compactenum}

\subsubsection{Adversary's Advantage.} The performance of the adversary is calculated as the probability of the adversary guessing the target sequence within the remaining pool of honey sequences.

\subsubsection{Sparse SNVs from the Target Sequence}
Figure~\ref{fig:low_rnd_seqs} illustrates how the log number of candidate sequences decreases with more side information available. With 1\% side information (10 SNVs), the number of sequences that match the side information reduces to approximately 44 on average across the three populations. 
Figure~\ref{fig:low_rnd} shows the increase of the adversary's advantage, averaged over 1000 rounds, vis-\`a-vis the number of SNVs available to her. 2.5\% side information (25 SNVs)  gives the adversary an advantage of approximately 80\% on average for the ASW and CEU populations and close to 90\% for the CHB population. With more side information, the adversary's advantage increases to over 90\%  for all populations.

\begin{figure*}[t]
  \centering
  \begin{subfigure}[b]{1.03\columnwidth}
  \centering
\includegraphics[width=0.835\columnwidth]{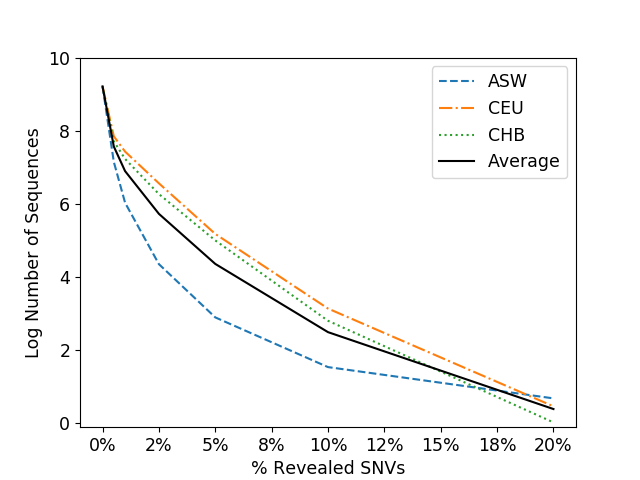} 
  \caption{\textmd{\#Candidate sequences vs \% revealed consecutive SNVs from target sequence}}
  \label{fig:low_consec_seqs}
	\end{subfigure}
	\hfill
  \begin{subfigure}[b]{1.03\columnwidth}
  \centering\includegraphics[width=0.835\columnwidth]{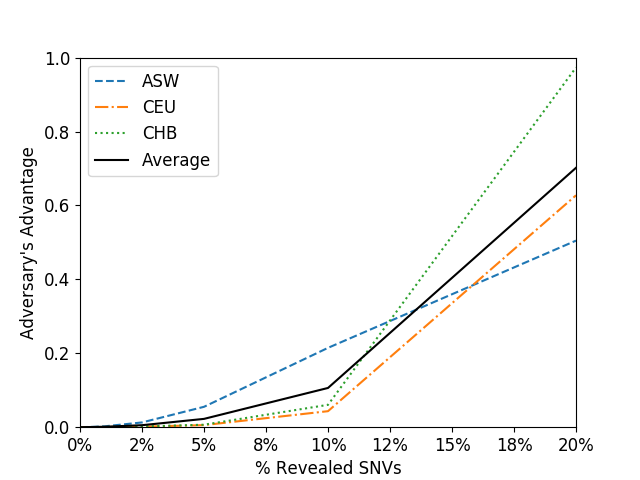} 
  \caption{\textmd{Adv's advantage vs \% revealed consecutive SNVs from target sequence}}
  \label{fig:low_consec}
  	\end{subfigure}
  \vspace{-0.3cm}
  \caption{Results of our evaluation in the low-entropy setting vis-\`a-vis an adversary with access to side information in the form of a cluster of consecutive SNVs from the target sequence.}
  \vspace{-0.1cm}
\end{figure*}

\begin{figure*}[t]
  \centering
	\hspace*{-0.5cm}
   \begin{subfigure}[b]{0.35\textwidth}
\includegraphics[width=0.99\columnwidth]{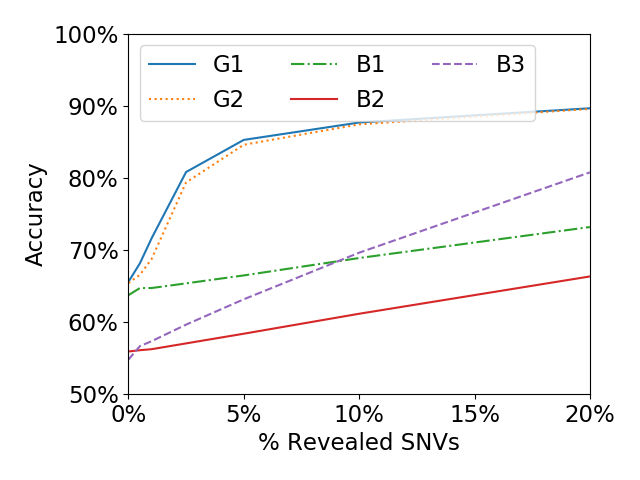} \vspace{-0.3cm}
  \caption{\textmd{ASW}}
  \label{fig:ASW_rnd}
  	\end{subfigure}
	\hspace*{-0.5cm}
   \begin{subfigure}[b]{0.35\textwidth}
\includegraphics[width=0.99\columnwidth]{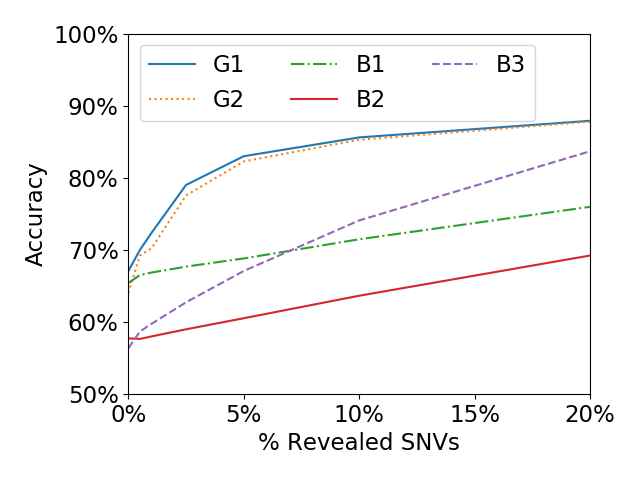}\vspace{-0.3cm}
  \caption{\textmd{CEU}}
   \label{fig:CEU_rnd}
  	\end{subfigure}
		\hspace*{-0.5cm}
   \begin{subfigure}[b]{0.35\textwidth}
\includegraphics[width=0.99\columnwidth]{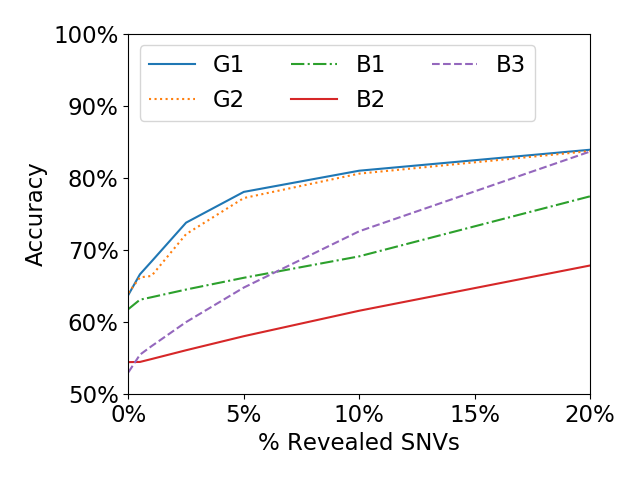} \vspace{-0.3cm}
    \caption{\textmd{CHB}}
    \label{fig:CHB_rnd}
  	\end{subfigure}
\vspace{-0.3cm}
    \caption{Inference accuracy results in the high-entropy password setting for all three populations for side information available to the attacker in the form of sparse SNVs from the target sequence.}
\label{fig:sparse_high}
\vspace{-0.3cm}
\end{figure*}

\subsubsection{Consecutive SNVs form the Target Sequence}
When the adversary has access to side information as a cluster of consecutive SNVs, she needs more side information to achieve comparable results to the Sparse SNVs case. 
Figure~\ref{fig:low_consec_seqs} shows the decrease of the log number of candidate sequences with increasing side information available. We observe the fastest decrease in the number of sequences with increasing side information available is for the ASW population when less than 10\% of the sequence available. Figure~\ref{fig:low_consec} shows the increase of the adversary's advantage, averaged over 1000 rounds, vis-\`a-vis the number of SNVs available to her. The increase in the adversary's advantage is slower as well, with an average of 70\% across the three populations for 20\% of the sequence available to the attacker.

\subsection{High-Entropy Password}\label{subsec:high-results}

\subsubsection{Experiment Overview}\label{subsec:overview}
The brute-force experiment presented in GenoGuard %
indicates that, when decrypting the same ciphertext with multiple passwords, the correct sequence would be ``buried'' among the incorrect ones. 
Hence, there is some similarity between the original sequence and the honey sequences.

As a result, we set to quantify the corresponding privacy loss, i.e. {\bf\em how much more information does an adversary obtain via access to ciphertext encrypted using GenoGuard  obtains, compared to one that was not}. 

Overall, we use the following evaluation strategy:
\begin{compactenum}
\item Encrypt a sequence using GenoGuard's DTE-then-encrypt method: for each of the 150 sequences, we select and encrypt 1,000 positions from chromosome 13, with a storage overhead $h = 4$, using a random, high-entropy password (approx.~72 bits).
\item Decrypt the ciphertext, using the top 10,000 most common passwords released by Daniel Miessler, to obtain plausible looking honey sequences;
\item Infer the victim's sequence using the honey sequences.
\end{compactenum}

\subsubsection{Accuracy} To measure the performance and assess the potential leakage that access to the GenoGuard ciphertext might yield, we measure the accuracy as the number of correctly guessed SNVs over the total number or SNVs guessed. 

\begin{figure*}
    \begin{minipage}{0.99\columnwidth}
 \centering
 \includegraphics[width=0.85\columnwidth]{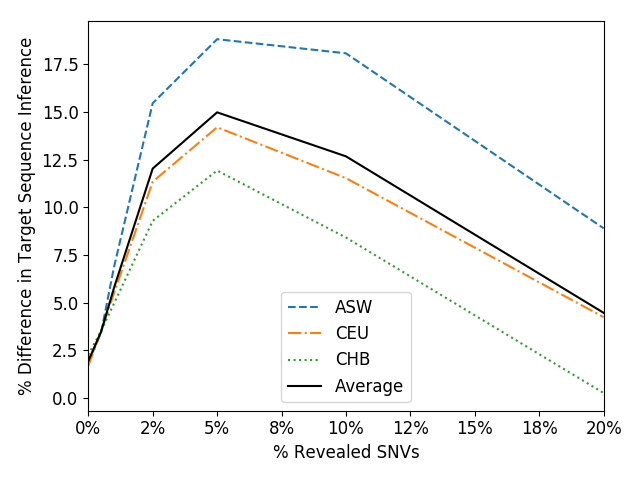}
\vspace{-0.3cm}
 \caption{Difference in accuracy between the best performing GenoGuard and baseline inference methods, vis-\`a-vis an adversary with side information of sparse SNVs from the target sequence, in the high-entropy password setting.}
  \label{fig:delta random}
\end{minipage}
\hfill
\setcounter{figure}{12}
   \begin{minipage}{0.99\columnwidth}
 \centering
 \includegraphics[width=0.9\columnwidth]{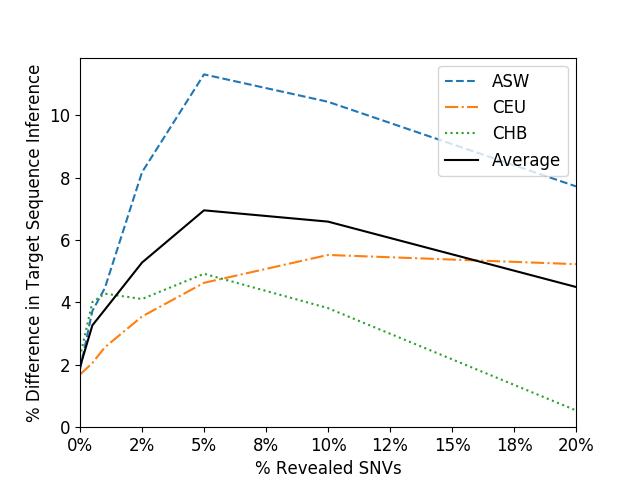}
  \vspace*{-0.2cm}
 \caption{
Difference in accuracy between the best performing GenoGuard and baseline inference methods, vis-\`a-vis an adversary with side information in the form of consecutive SNVs from the target sequence, in the high-entropy password setting.} 
  \label{fig:delta_consec}
\end{minipage}
\end{figure*}

\setcounter{figure}{11}
\begin{figure*}[t]
  \centering
  \hspace*{-0.5cm}
   \begin{subfigure}[b]{0.35\textwidth}
\includegraphics[width=0.99\columnwidth]{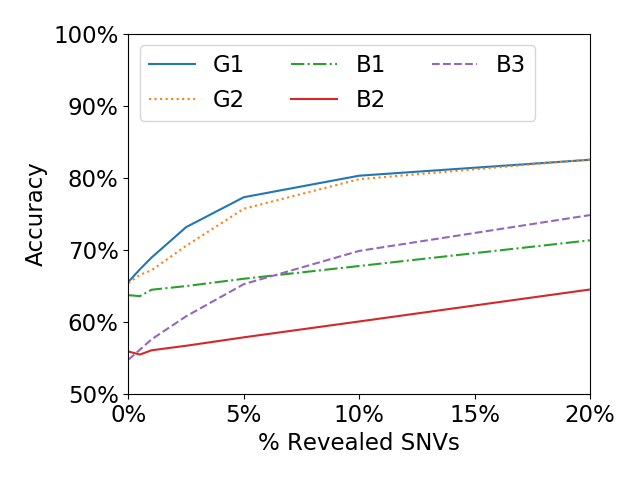} \vspace{-0.3cm}
  \caption{\textmd{ASW}}
  \label{fig:ASW_consec}
  	\end{subfigure}
  \hspace*{-0.5cm}
   \begin{subfigure}[b]{0.35\textwidth}
\includegraphics[width=0.99\columnwidth]{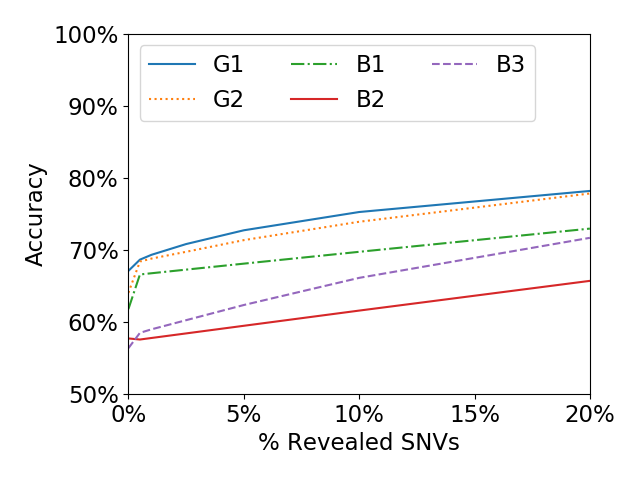}\vspace{-0.3cm}
  \caption{\textmd{CEU}}
  \label{fig:CEU_consec}
  	\end{subfigure}
  \hspace*{-0.5cm}
   \begin{subfigure}[b]{0.35\textwidth}
\includegraphics[width=0.99\columnwidth]{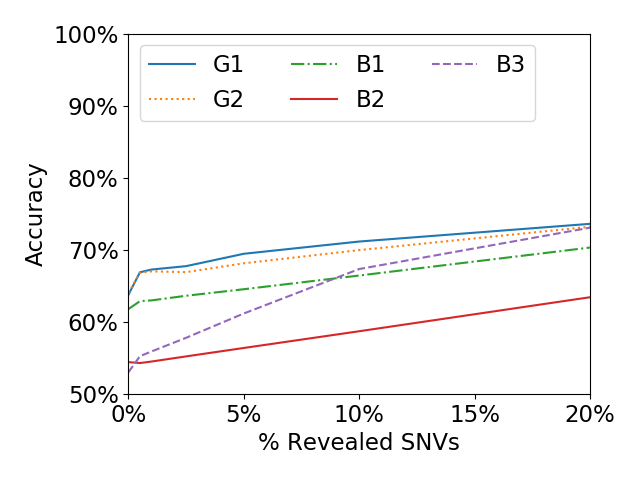}
\vspace{-0.3cm}
    \caption{\textmd{CHB}}
    \label{fig:CHB_consec}
  	\end{subfigure}
\vspace{-0.3cm}
    \caption{Inference accuracy results in the high-entropy password setting for all three populations for side information available to the attacker in the form of a cluster of consecutive SNVs from the target sequence.}
\label{fig:consec_high}
\vspace{-0.2cm}
\end{figure*}

\subsubsection{Sparse SNVs from the Target Sequence}

Figure~\ref{fig:sparse_high} shows the inference results in this case for the three population groups, averaged over 1,000 rounds. In the case where no side information is available to the attacker, for all populations, the attacker can infer approximately 2\% more of the target sequence from the GenoGuard ciphertext than just by using baseline inferences based on the population.
For the ASW population (Figure~\ref{fig:ASW_rnd}), over 80\% of the target SNVs are guessed correctly with 2.5\% (25 SNVs) or more of the target sequence available to the attacker.  For the CEU population (Figure~\ref{fig:CEU_rnd}), approximately 79\% of the target SNVs are guessed correctly with 2.5\% of the original sequence available to the attacker  and over 83\% of the target SNVs are guessed correctly with 5\% (50 SNVs) or more are available. In the case of the CHB population (Figure~\ref{fig:CHB_rnd}), the accuracy is of the GenoGuard inference is the lowest among the three populations, with over 73\% accuracy in the cases where 2.5\% SNVs are available to the attacker. The accuracy surpasses 80\% for the CHB population when 10\% or more of the target SNVs are available to the attacker.

In Figure~\ref{fig:delta random}, we illustrate the difference between the best performing inference method using the GenoGuard ciphertext and the best performing baseline inference method. On average, having access to the GenoGuard ciphertext improves the inference accuracy. The peak of the improvement in accuracy (approximately 15\%) over the baseline models can be observed when the attacker has access to 5\% sparse SNVs from the target sequence. After this, we can see a decline in this difference with increasing SNVs available for the attacker, as the baseline inference becomes more accurate with more information available. In fact, for the CHB population, the best performing baseline (B3) for the case when 20\% of the target sequence is available to the attacker provides an accuracy comparable to the GenoGuard inferences ($\approx$83.8$\%$).

\subsubsection{Consecutive SNVs form the Target Sequence}
In Figure~\ref{fig:consec_high}, we illustrate the accuracy of the inference methods across the three populations when the adversary obtains, as side information, a cluster of consecutive SNVs, averaged over 1,000 rounds.
For the ASW population (Figure~\ref{fig:ASW_consec}), the accuracy of inferred SNVs from the correct sequence using the GenoGuard ciphertext is over 73\% for 2.5\% or more of the target SNVs available as side information, and over 80\% when 10\% or more of the sequence is available to the attacker. The GenoGuard inference for the CEU population (Figure~\ref{fig:CEU_consec}) is over 70\% when 2.5\% or more of the target sequence is available to the attacker. For the CHB population (Figure~\ref{fig:CHB_consec}), the GenoGuard inferences have the lowest accuracy across the three populations, obtaining 70\% or more accuracy only when 5\% or more of the target sequence is available to the attacker.

Figure~\ref{fig:delta_consec} shows the difference between the best performing GenoGuard inference method and the best performing baseline inference method. On average, the inference using the GenoGuard ciphertext gives better accuracy than the baseline methods, but overall less than the previous case where sparse SNVs are available as side information. In this case, the peak in the improvement of accuracy compared to the baseline methods is approximately 7\%, on average, across the three populations, when 5\% of the target SNVs are available to the attacker.
For the CHB population, when 20\% of the sequence is available as side information to the attacker, we observe, as in the case of sparse SNVs, that the best performing baseline inference method (B3) obtains a comparable accuracy to that of the GenoGuard inferences ($\approx$73$\%$).
 
\subsection{Take-Aways}

Overall, our experimental evaluation shows that, when the adversary has access to some side information, access to a ciphertext encrypted using GenoGuard can help her recover a remarkably high percentage of the SNVs from the target sequence or significantly increase her advantage in recovering the correct sequence.

Therefore, users need to include as much side information as possible when encrypting their genomic sequence. However, this prompts a parallel problem, with respect to how much that user is willing to publicly share (as this information is saved together with the ciphertext), considering that even without access to the GenoGuard ciphertext, it can enable attackers to correctly predict most of the target genome.
 
\section{Related Work}\label{sec:related work}
In this section, we review relevant prior work on genome privacy and honey encryption.

\subsection{Genome Privacy}

\descr{Re-identification.} Genomic data is hard to anonymize, due to the genome's uniqueness as well as correlations within different regions.
For instance, Gymrek et al.~\cite{gymrek_identifying_2013} demonstrate that surnames of genomic data donors can be inferred using data publicly available from recreational genealogy databases. %
They also discuss how, through deep genealogical ties, publishing even a few markers can lead to the identification of another person who might have no acquaintance with the one who released their genetic data. %
In follow-up work, Erlich et al.~\cite{Erlich690} show that a genetic database which covers only 2\% of the target population can be used to find a third-cousin of nearly any individual.

\descr{Membership inference.} Homer et al.~\cite{homer_resolving_2008} present a membership inference attack in which they infer the presence of an individual's genotype within a complex genomic DNA mixture.
Wang et al.~\cite{wang2009learning} improve on the attack using correlation statistics of just a few hundreds SNPs, while Im et al.~\cite{im2012sharing} rely on regression coefficients.
Shringarpure and Bustamante~\cite{beacon_SB} perform membership inference against the Beacon network.\footnote{Beacons are web servers that answer questions e.g. ``does your dataset include a genome that has a specific nucleotide at a specific genomic coordinate?'' to which the Beacon responds yes or no, without referring to a specific individual; see: \url{https://github.com/ga4gh-beacon/specification}.}
They use a likelihood-ratio test to predict whether an individual is present in the Beacon, detecting membership within a Beacon with 1,000 individuals using 5,000 queries. %
Also, Von Thenen et al.~\cite{von_Thenen200147} reduce the number of queries to less than 0.5\%. 
Their best performing attack uses a high-order Markov chain to model the SNP correlations, as described in~\cite{samani_quantifying_2015}.
Note that, as part of the attacks described in this paper, we use inference methods from~\cite{samani_quantifying_2015} as our baseline inference methods.

\descr{Data sharing.} Progress in genomics research is dependent on collaboration and data sharing among different institutions.
Given the sensitive nature of the data, as well as regulatory and ethics constraints, this often proves to be a challenging task.
Kamm et al.~\cite{kamm_new_2013} propose the use of secret sharing to distribute data among several entities and, using secure multi-party computations, support privacy-friendly computations across multiple entities.
Wang et al.~\cite{Wang2015} present GENSETS, a genome-wide, privacy-preserving similar patients querying system using genomic edit distance approximation and private set difference protocols. 
Then, Chen et al.~\cite{chen_princess:_2017} use Software Guard Extensions (SGX) to build a privacy-preserving international collaboration tool; this enables secure and distributed computations over encrypted data, thus supporting the analysis of rare disease genetic data across different continents.
Finally, Oprisanu and De Cristofaro~\cite{oprisanu2018anonimme} present a framework (``AnoniMME'') geared supporting anonymous queries within the Matchmaker Exchange platform, which allows researchers to perform queries for rare genetic disease discovery over multiple federated databases.

\descr{Privacy-friendly testing.} Another line of work focuses on protecting privacy in the context of personal genomic testing, i.e., computational tests run on sequenced genomes to assess, e.g., genetic susceptibility to diseases, determining the best course of treatment, etc. 
Baldi et al.~\cite{baldi2011countering} assume that each individual keeps a copy of their data and consents to tests done in such a way that only the outcome is disclosed. 
They present a few cryptographic protocols allowing researchers to privately search mutations in specific genes.
Ayday et al.~\cite{ayday_protecting_2013} rely on a semi-trusted party to store an encrypted copy of the individual's genomic data: using additively homomorphic encryption and proxy re-encryption, they allow a Medical Center to privately perform disease susceptibility tests on patients' SNPs. 
Naveed et al.~\cite{naveed14} introduce a new cryptographic primitive called Controlled Functional Encryption (CFE), 
which allows users to learn only certain functions of the (encrypted) data, using keys obtained from an authority; however, the client is required to send a fresh key request to the authority every time they want to evaluate a function on a ciphertext.
Overall, for an overview of privacy-enhancing technologies applied to genetic testing, we refer the reader to~\cite{sok}.

\descr{Long-term security.} %
As the sensitivity of genomic data does not degrade over time, access to an individual's genome poses a threat to her descendants, even years after she has deceased. 
To the best of our knowledge, GenoGuard~\cite{huang_genoguard:_2015} is the only attempt %
to provide long-term security. 
GenoGuard, reviewed in Section~\ref{sec:genoguard}, relies on Honey Encryption~\cite{HE2}, aiming to provide confidentiality in the presence of brute-force attacks;
it only serves as a storage mechanism, i.e., it does not support selective retrieval or testing on encrypted data (as such, it is not ``composable'' with other techniques supporting privacy-preserving testing or data sharing). 
In this paper, we provide a security analysis of GenoGuard.
In parallel to our work, Cheng et al.~\cite{cheng} recently propose attacks against probability model transforming encoders, and also evaluate them on GenoGuard. 
Using machine learning, they train a classifier to distinguish between the real and the decoy sequences, and exclude all decoy data for approximately 48\% of the individuals in the tested dataset.

\subsection{Honey Encryption}
Juels and Ristenpart~\cite{HE2} %
introduce Honey Encryption (HE) as a general approach to encrypt messages using low min-entropy keys such as passwords.
HE, reviewed in Section~\ref{sec:he},  is designed to yield plausible-looking ciphertexts, called honey messages, even when decrypted with a wrong password. 
In a nutshell, it uses a distribution-transforming-encoder (DTE) to encode a-priori knowledge of the message distribution, 
aiming to provide {\em message recovery} security against computationally unbounded adversaries.
It was originally designed to encrypt credit card information, RSA secret keys, etc.~\cite{tyagi2015honey}. 

Message recovery security can be defined as follows~\cite{jaeger2016honey}: given a message encrypted under a key whose maximum probability of taking on any particular value is at most $1/2^\mu$, an unbounded adversary's ability to guess the correct message, even given the ciphertext, is at most $1/2^\mu$ plus a negligible amount.
However, Jaeger et al.~\cite{jaeger2016honey} discuss deficiencies of message recovery security as per modern security goals. 
More specifically, not only they prove the impossibility of known-message attack security in the case of low-entropy keys, but also mention that schemes meeting message recovery security might actually leak a significant amount of information about the plaintexts, even if the adversary cannot correctly recover the full message with non-negligible probability.
Although this serves as an inspiration to our work, note that the context of our evaluation is different, as in the low-entropy setting, we show that a lower bound also applies to the adversary's advantage when partial information from the target sequence is available to the attacker, compared to having pairs of ciphertext and plaintext.
Another work studying attacks against HE is that by Cheng et al.~\cite{cheng}, which we have reviewed above.

\descr{Honeywords.} Before Honey Encryption~\cite{HE2}, Juels and Rivest~\cite{juels2013honeywords} introduced the concept of ``honeywords'' to improve the security of password databases. 
They propose adding honeywords (false passwords) to a password database together with the actual password (hashed with salt) of each user.
This way, an adversary who hacks into the password database and inverts the hash function cannot know whether she has found the password or a honeyword.

Wang et al.~\cite{wang2018security} present an evaluation of the honeyword system~\cite{juels2013honeywords}, finding it to be vulnerable to a number of attacks. %
More specifically, an adversary that wants to distinguish between real and decoy passwords can do so with a success rate of 30\% compared to an expected 5\%. In the case of a targeted attack, when the adversary is assumed to know some personal information about the user, they show that the adversary's success rate is further improved to about 60\%. %
Our attacks differ from those in~\cite{wang2018security}, first, as they target the honeywords system~\cite{juels2013honeywords}, while we focus on Honey Encryption~\cite{HE2}, and in particular its application to GenoGuard~\cite{huang_genoguard:_2015}. 
Moreover, their attack only aims to identify the correct password from a given password pool, while we also examine the case when the correct password is not found within the tried passwords.

\section{Conclusion} \label{sec:conc}
 
Motivated by the decreasing cost of genomic sequencing and the related arising privacy challenges, the research community has produced a large body of work  on genomic privacy. Most of the techniques focus on cryptographic tools, but fail to address the need for long term confidentiality for genomic data. %
In fact, GenoGuard~\cite{huang_genoguard:_2015} is the only tool available for ensuring the long term encryption needed for genomic data \cite{sok}. 
 
In this paper, we set to determine whether GenoGuard can be safely used as an encryption tool, quantifying the additional privacy leakage arising from using it. We analyzed GenoGuard under two scenarios, based on the encryption password, for an adversary which has access to side information about the target sequence in the form of some values of SNVs from the target sequence. First, we assumed that the user encrypts his genomic sequence using a low-entropy, easily guessable password. In this case, we found that the adversary can easily exclude decoy passwords from the pool of possible passwords, and can guess the correct sequence with high probability by having access to 2.5\% sparse SNVs or 20\% or more consecutive SNVs from the target sequence. 

Second, we assumed that the user encrypts his sequence using a high-entropy password. In this case, since elimination of decoy passwords might not yield any sequence, we use the honey sequences to obtain as much information as possible from the target sequence, exploiting the similarity between the original sequence and the honey sequences \cite{huang_genoguard:_2015}. We then compared the sequence obtained from the honey sequences to state-of-the-art methods from genome sequence inference in order to observe the privacy leakage. Even with no side information available to the attacker, the sequence obtained from the honey sequences had a 2\% improvement on average over all tested baseline methods. With side information in the form of sparse SNVs from the target sequence, the improvement in accuracy compared to the baseline inference models raises to up to 15\% on average when 5\% of the target sequence is available to the attacker, predicting more than 82\% (on average) of the target sequence correctly. When the attacker obtains consecutive SNVs from the target sequence, the accuracy of the attacker decreases slightly from the previous case, yielding 73\% accuracy when 5\% of the target sequence is known, with an average improvement of 7\% over the baseline methods.

In conclusion, we argue that the research community should invest more resources toward the design of long-term encryption tools for genomic data. 
Overall, GenoGuard could be a viable solution when the user incorporates {\em all} side information into the encryption. 
However, given the fact that all this information needs to be stored together with the ciphertext, it also prompts the question of how much is a user willing to disclose, considering that only the baseline methods can predict, with high accuracy, the correct sequence (e.g. with 20\% sparse SNVs available to the attacker, her accuracy is, on average, over 82\%). 
Users who have already used GenoGuard for long-term encryption purposes need to be aware that if further genomic information can be obtained by the attacker, it will severely diminish the security of the system. 

As part of future work, we plan to analyze the security of GenoGuard for side information arising from kinship associations. 

\descr{Acknowledgments.} This work was supported by a Google Faculty Award on ``Enabling Progress in Genomic Research via Privacy Preserving Data Sharing,'' the European Union's Horizon 2020 Research and Innovation program under the Marie Sk\l{}odowska-Curie ``Privacy\&Us'' project (GA No. 675730), and the Swiss National Science Foundation (Grant 150654).

\balance
\bibliographystyle{abbrv}

\end{document}